\documentclass[12pt]{article}
\usepackage{amssymb}
\usepackage{amsmath}
\usepackage{mathtools}
\usepackage{geometry}
\usepackage{setspace}
\usepackage{natbib}
\usepackage{indentfirst}
\usepackage{array,epsfig,fancyheadings,rotating}
\usepackage{verbatim}
\usepackage{sectsty, secdot}
\usepackage[]{hyperref}
\usepackage{caption}
\usepackage{enumitem}
\usepackage{url}
\usepackage{color}
\newtheorem{theorem}{Theorem}[section]

\newtheorem{lemma}{Lemma}[section]

\newenvironment{proof}[1][Proof]{\noindent\textbf{#1.} }{\ \rule{0.5em}{0.5em}}
\sectionfont{\fontsize{12}{14pt plus.8pt minus .6pt}\selectfont}

\subsectionfont{\fontsize{12}{14pt plus.8pt minus .6pt}\selectfont}
\begin{document}

\title{A Broad and General Sequential Sampling Scheme} \author{ Jun Hu\footnote{Jun Hu is an Assistant Professor in the Department of Mathematics and Statistics, Oakland University, 146 Library Drive, Rochester, MI 48309, USA. Tel.:~248-370-3434. Email address: junhu@oakland.edu.} \; and Yan Zhuang\footnote{Yan Zhuang is an Assistant Professor in the Mathematics and Statistics Department, Connecticut College. Email address: yzhuang@conncoll.edu.} \; }

\date{}

\maketitle

\bigskip

\begin{abstract}

In this paper, we propose a broad and general sequential sampling scheme, which incorporates four different types of sampling procedures: i) the classic Anscombe-Chow-Robbins purely sequential sampling procedure; ii) the ordinary accelerated sequential sampling procedure; iii) the relatively new $k$-at-a-time purely sequential sampling procedure; iv) the new $k$-at-a-time improved accelerated sequential sampling procedure. The first-order and second-order properties of this general sequential sampling scheme are fully investigated with two illustrations on minimum risk point estimation for the mean of a normal distribution and on bounded variance point estimation for the location parameter of a negative exponential distribution, respectively. We also provide extensive computational simulation studies and real data analyses for each illustration.

\bigskip


\bigskip

\end{abstract}

\bigskip


\setcounter{section}{0}
\setcounter{equation}{0}
\section{Introduction}\label{Sect. 1}

In statistical inference problems where there exists no
fixed-sample-size procedure, sequential sampling schemes have been developed and widely used with proved efficiency properties in terms of the sample size required. The fundamental theories of sequential estimation are largely
based on ground-breaking papers due to \cite{Anscombe (1953)} and \cite{Chow and Robbins (1965)}, in which purely sequential sampling methodologies were developed for the problem of constructing fixed-width confidence intervals.

In a parallel path, \cite{Robbins (1959)} originally formulated the minimum risk point estimation (MRPE) problem. Under the absolute error loss plus linear cost of sampling, a purely sequential stopping rule was proposed to estimate an unknown normal mean $\mu$ when the variance $\sigma^{2}$ was assumed unknown. Then, \cite{Starr (1966)} and \cite{Starr and Woodroofe (1969)} considered a more general loss function and proved a number of interesting asymptotic first-order and second-order properties of the purely sequential MRPE methodology. Using nonlinear renewal theoretical tools developed in \citeauthor{Lai and Siegmund (1977)} (\citeyear{Lai and Siegmund (1977)}, \citeyear{Lai and Siegmund (1979)}), \cite{Woodroofe (1977)} further developed explicit second-order approximations associated with \textit{efficiency}, \textit{risk efficiency} and \textit{regret}. Moreover, \cite{Ghosh and Mukhopadhyay (1980)} provided a different method to evaluate the expression for regret, which could generalize the corresponding result by \cite{Woodroofe (1977)}.

Let us begin with a sequence of i.i.d. positive and continuous random variables $\{W_{n},n\ge 1\}$. For simplicity, we assume that these random variables have all positive moments finite, with mean $E[W_{1}]=\theta$ and
variance $V[W_{1}]=\tau^{2}$ specified in particular. In addition, the distribution function of $W_{1}$ satisfies the following condition:
\begin{equation*}
P(W_{1}\le x)\le Bx^{\alpha} \text{ for all } x>0
\end{equation*}
and for some $B>0$ and $\alpha>0$, both free from $x$. Similar with (1.1) of \cite{Woodroofe (1977)}, all the stopping times arising from the above inference problems can be written in a general form given by
\begin{equation}\label{1.1}
t_{0}=\inf\left\{n\ge m:n^{-1}\sum_{i=1}^{n}W_{i}\le \theta (n/n^{\ast})^{\delta}l_{1}(n)\right\},
\end{equation}
where $\delta>0$ is a positive constant, $l_{1}(n)=1+l_{0}n^{-1}+o(n^{-1})$ as $n\rightarrow \infty$ with $-\infty <l_{0}<\infty$ is a convergent sequence of numbers, $m \ge 1$ indicates a pilot sample size, and $n^{\ast}$ is called an {\it optimal fixed sample size} to be determined in specific problems. One may refer to Section A.4 of the Appendix in \cite{Mukhopadhyay and de Silva (2009)} for more details.

The stopping rule \eqref{1.1} is implemented as follows: After an initial sample of size $m$ is gathered, one observation is taken one-at-a-time as needed successively. Each time when a new observation is recorded, one evaluates
the sample data to check with the stopping rule and terminates sampling at the first time that the stopping rule is satisfied. Therefore, this is a purely sequential sampling scheme and let us denote it by $\mathcal{M}_{0}$.

To introduce the properties of the stopping time $t_0$ and relate its expected value with the optimal fixed sample size, let us define a general function of $\eta(k)$ as follows. For each integer $k\ge 1$, 
\begin{equation}\label{1.2}
\eta(k)=\frac{k}{2}-\frac{1}{2}\delta^{-1}\theta^{2}\tau^{2}-\delta^{-1}l_{0}-(\delta\theta)^{-1}\sum_{n=1}^{\infty}n^{-1}E\left[\left\{\sum_{i=1}^{kn}W_{i}-kn(\delta+1)\theta\right\}^{+}\right],
\end{equation}
where $\{u\}^{+}=\max\{0,u\}$. Under certain conditions, \cite{Woodroofe (1977)} fully studied the properties of $t_{0}$, which are summarized in the following theorem.
\begin{theorem}\label{Theorem 1.1}
For the purely sequential sampling scheme $\mathcal{M}_{0}$ and the stopping time $t_{0}$ given in \eqref{1.1}, we have as $n^{\ast}\rightarrow \infty $: If $m>(\alpha\delta)^{-1}$,
\begin{equation*}\label{1.3}
E_{\theta,\tau }\left[t_{0}-n^{\ast }\right] =\eta (1)+o(1).
\end{equation*}
\end{theorem}

In the spirit of \cite{Hall (1983)}, we define the term of sampling operation as the procedure of collecting new observations and evaluating the sample data to make a decision. Let $\varphi$ denote the number of sampling operations. For the purely sequential sampling scheme $\mathcal{M}_{0}$ associated with the stopping time $t_{0}$ given in \eqref{1.1}, we have
\begin{equation*}\label{1.4}
\varphi_{\mathcal{M}_{0}}=t_{0}-m+1,
\end{equation*}
and
\begin{equation}\label{1.5}
E_{\theta,\tau}[\varphi_{\mathcal{M}_{0}}]=n^{\ast}+\eta(1)-m+1+o(1).
\end{equation}

Not surprisingly, the purely sequential sampling scheme $\mathcal{M}_{0}$ requires a lot of sampling operations. In view of this, \cite{Hall (1983)} proposed an accelerated sequential estimation methodology saving sampling operations. Furthermore, \cite{Mukhopadhyay and Solanky (1991)} and \cite{Mukhopadhyay (1996)} developed alternative formulations of the accelerated sequential sampling technique. \cite{Liu (1997)} improved the purely sequential sampling scheme of Anscombe-Chow-Robbins and proposed a new sequential methodology requiring substantially fewer sampling operations. Accelerated sequential sampling methodologies first draw samples sequentially part of the way and then augment with sampling in one single batch. As per the discussions from \citet[p. 228]{Mukhopadhyay and de Silva (2009)}, ``An accelerated sequential strategy would always be operationally much more convenient in practical implementation than its sequential counterpart!''

On the other hand, \cite{Hayre (1985)} considered sampling in bulk or groups, rather than one at a time, and proposed group sequential sampling with variable group sizes.  \cite{Schmegner and Baron (2004)} presented a concept of \textit{sequential planning} as extension and generalization of purely sequential procedures. Implementation of these sampling schemes proved to require only a few groups to cross the stopping boundary, leading to only a moderate increase in sample size.

Most recently, \cite{Mukhopadhyay and Wang (2020a)} first brought up a new type of sequential sampling scheme in which they considered recording $k$ observations at a time, given the thoughts that in real life packaged items purchased in bulk often cost less per unit sample than the cost of an individual item. They discussed this new sampling strategy in both FWCI and MRPE problems for the mean of a normal population. In \cite{Mukhopadhyay and Wang (2020b)}, they revisited these problems and incorporated the newly constructed estimators under permutations within each group for the stopping boundaries, which led to tighter estimation of required sample sizes. \cite{Hu (2020)} developed a double-sequential sampling scheme which is defined as $k$-at-a-time part of way, and then one-at-a-time sequentially, which requires similar sample sizes as the purely sequential strategies but saves sampling operations. \cite{Mukhopadhyay and Sengupta (2021)} further proposed sequential estimation strategies for big data science with minimal computational complexities, where they proposed the idea of $k$-tuples instead of one single observation.  

In this paper, we incorporate these path-breaking ideas with modification to accelerate the purely sequential sampling scheme without sacrificing the first-order and second-order efficiency: (i) using the purely sequential methodology to determine only a proportion $\rho(0<\rho<1)$ of the desired final sample, and then augment with sampling in one single batch; and (ii) drawing a fixed number $k(k\ge 2)$ observations at-a-time successively until termination in the sequential sampling portion. In this way, we expect to save roughly $100(1-k^{-1}\rho)\%$ of sampling operations. Allowing both $\rho$ and $k$ to be $1$, we can therefore propose a new and general sequential sampling scheme, denoted by $\mathcal{M}(\rho,k)$, in Section \ref{Sect. 2} along with a number of desirable properties. The sampling scheme of $\mathcal{M}(\rho,k)$ provides a wide range of sampling procedures with different selections of $k$ and $\rho$ values: 
\begin{itemize}
	\item[(i)] $k=1,\rho=1$, the classic Anscombe-Chow-Robbins purely sequential procedure;
	\item[(ii)] $k=1,0<\rho<1$, the ordinary accelerated sequential procedure;
	\item[(iii)] $k>1,\rho=1$, the $k$-at-a-time purely sequential procedure;
	\item[(iv)] $k\ge2,0<\rho<1$, the $k$-at-a-time improved accelerated sequential procedure.
\end{itemize}

The rest of this paper is organized as follows. Section 2 proposes the general sequential sampling scheme $\mathcal{M}(\rho,k)$ and explores its appealing first-order and second-order properties with a special focus on the $k$-at-a-time improved accelerated sequential procedure. In Section \ref{Sect. 3}, we construct MRPE for an unknown normal mean $\mu$ with the variance $\sigma^{2}$ also assumed unknown as a possible illustration of the newly proposed general sequential sampling scheme $\mathcal{M}(\rho,k)$. In Section \ref{Sect. 4}, we construct bounded variance point estimation for an unknown location parameter $\mu$ from a negative exponential distribution. Simulated performances and real data analysis are included to support and supplement our theory for both methodologies in Section \ref{Sect. 3} and Section \ref{Sect. 4}. Section \ref{Sect. 5} shares some brief concluding thoughts.


\setcounter{equation}{0}
\section{The General Sequential Sampling Scheme $\mathcal{M}(\rho,k)$}\label{Sect. 2}

Under the same assumptions in Section \ref{Sect. 1}, we propose a broader and more general sequential sampling scheme $\mathcal{M}(\rho,k)$ associated with the following stopping times modified in view of \eqref{1.1}: 
\begin{equation}\label{2.1}
\begin{split}
t_{1} &\equiv t_{1}(\rho,k)=\inf\left\{n\ge m:(kn)^{-1}\sum_{i=1}^{n}U_{i}\le \theta \lbrack kn/(\rho n^{\ast})]^{\delta}l_{k}(n)\right\},\\ 
t_{2} &\equiv t_{2}(\rho,k)=\left\lfloor \rho^{-1}kt_{1}(\rho
,k)\right\rfloor + 1.
\end{split}
\end{equation}
In addition to the notation in the stopping rule \eqref{1.1}, $k\ge 1$ is a prefixed integer, $0<\rho \le 1$ is a prefixed proportion, $l_{k}(n)=1+l_{0}(kn)^{-1}+o(n^{-1})$ as $n\rightarrow \infty $ with $-\infty<l_{0}<\infty $ is a convergent sequence of numbers, $U_{i}\equiv \sum_{j=(i-1)k+1}^{ik}W_{j},i=1,2,...$ are i.i.d. random variables, and $\left\lfloor u\right\rfloor$ denotes the largest integer that is strictly smaller than $u$ ($<u$).

The implementation of the stopping rule \eqref{2.1} can be interpreted as follows: Starting with $km$ pilot observations, we sample $k$ observations at-a-time as needed successively and determine a preliminary sample size of $kt_{1}(\rho,k)$. Then, we continue to sample $t_{2}(\rho,k)-kt_{1}(\rho,k)$ additional observations if needed all in one batch. Obviously, $P_{\theta,\tau}(t_{2}(\rho,k)<\infty)=1$ and $t_{2}(\rho,k)\uparrow \infty $ w.p.1 as $n^{\ast}\uparrow \infty $. If both $\rho$ and $k$ are chosen to be 1, then our newly developed sampling scheme $\mathcal{M}(1,1)$ will be the ordinary purely sequential sampling scheme $\mathcal{M}_{0}$ associated with the stopping rule \eqref{1.1}. If $\rho=1$ and $k\ge 2$, the new sampling scheme $\mathcal{M}(1,k)$ is purely sequential, but taking multiple ($k$) observations at-a-time. If $0<\rho <1$ and $k=1,$ the new sampling scheme $\mathcal{M}(\rho,1)$ is an ordinary accelerated sequential sampling scheme, similar with the one proposed in \cite{Mukhopadhyay (1996)}.

Compared with the purely sequential sampling scheme $\mathcal{M}_{0}$, our newly developed sequential sampling scheme $\mathcal{M}(\rho,k)$ can reduce approximately $100(1-k^{-1}\rho)\%$ of the operational time, which makes it flexible in real practice. One is allowed to choose the values of $k$ and $\rho$ to optimize the sampling process, time limitations, and cost considerations under different situations. We incorporate a brief discussion here to illustrate the flexibility of our sampling scheme.

If the cost of taking more observations is high, one may determine a smaller $k$ value and/or choose $\rho$ to be closer to $1$; and if the process is more time sensitive, one may use the methodology with a larger $k$ and/or choose $\rho$ to be closer to 0. As a direct result of operational convenience, the new sequential sampling scheme $\mathcal{M}(0<\rho <1,k\ge 2)$ tends to oversample, but the increase in the expected sample size is bounded by an amount depending on $\rho$ and $k$. Moreover, we have a thorough investigation of the efficiency properties for the sampling scheme $\mathcal{M}(\rho,k)$ in the sequel.

To establish asymptotic properties, we suppose that the following limit operations hold in the spirit of \cite{Hall (1983)}:
\begin{equation}\label{2.2}
m\rightarrow \infty, \ n^{\ast}=O(m^{r}), \ \text{and} \ \lim\sup\frac{m}{n^{\ast}}<\rho,
\end{equation}
where $r>1$ is a fixed constant. In the light of \eqref{1.1} and Theorem \ref{Theorem 1.1}, we are now in a position to state the major results of this paper as Theorem \ref{Theorem 2.1}. See \citet[Theorem 2.4]{Woodroofe (1977)} for more details. 
\begin{theorem}\label{Theorem 2.1}
For the new general sequential sampling scheme $\mathcal{M}(\rho,k)$ and the stopping times $t_{1}(\rho,k)$ and $t_{2}(\rho,k)$ given in \eqref{2.1}, for fixed $0<\rho<1$ and $k$, with $\eta(k)$ defined in \eqref{1.2}, under the limit operations \eqref{2.2}:
\begin{equation*}
E_{\theta,\tau}\left[kt_{1}(\rho,k)-\rho n^{\ast}\right] = \eta(k)+o(1)
\end{equation*}
\begin{equation*}\label{2.3}
\rho^{-1}\eta(k)+o(1)\leq E_{\theta,\tau}\left[t_{2}(\rho,k)-n^{\ast}\right] \leq \rho ^{-1}\eta(k)+1+o(1).
\end{equation*}
And if $\rho=1$, then
$$E_{\theta,\tau}\left[t_{2}(1,k)-n^{\ast}\right]=E_{\theta,\tau}\left[kt_{1}(1,k)- n^{\ast}\right]=\eta(k)+o(1).$$
\end{theorem}
It is clear when $0<\rho<1$ and $k\ge 2$,
our new sampling scheme $\mathcal{M}(\rho,k)$ is expected to oversample up to $\rho^{-1}\eta(k)+1+o(1)$ observations. In terms of the number of sampling operations, it is not hard to obtain that
\begin{equation*}\label{2.4}
\varphi_{\mathcal{M}(\rho,k)}=t_{1}-m+1+I(\rho <1),
\end{equation*}
where $I(D)$ stands for the indicator function of an event $D$. We also have
\begin{equation}\label{2.5}
E_{\theta,\tau}[\varphi_{\mathcal{M}(\rho,k)}]=k^{-1}[\rho n^{\ast}+\eta(k)-km]+1+I(\rho <1)+o(1).
\end{equation}
Comparing \eqref{1.5} and \eqref{2.5}, the new general sequential sampling scheme $\mathcal{M}(\rho,k)$ requires roughly $100(1-k^{-1}\rho)\%$ fewer sampling operations than those of the purely sequential sampling scheme $\mathcal{M}_{0}$, based on the actual choices of $k$ and $\rho$. Therefore, it enjoys great operational convenience with a cost of only a slight increase in the projected final sample size. Starting from here and thereafter, we mainly focus on the sampling scheme $\mathcal{M}(\rho,k)$ with $0<\rho<1$ and $k \ge 2$ which makes it specifically the $k$-at-a-time improved accelerated sequential sampling scheme. Nevertheless, one may note that all the theories and methodologies we discuss actually work for the general sequential sampling scheme with $0<\rho\leq 1$ and/or $k\geq 1$.


\setcounter{equation}{0}
\section{Minimum Risk Point Estimation for a Normal Mean}\label{Sect. 3}

In this section, we discuss minimum risk point estimation (MRPE) for a normal mean as an illustration of our $k$-at-a-time improved accelerated sequential sampling scheme $\mathcal{M}(\rho,k)$. Having recorded a sequence of independent observations, $X_{1},...,X_{n}, n\ge 2$ from a $N(\mu,\sigma^{2})$ population, where both 
$\mu$ and $\sigma$ are unknown, $\mu \in \mathcal{R}$ and $\sigma \in \mathcal{R}^{+}$, we denote the sample mean and sample variance as follows:
\begin{equation*}
\begin{tabular}{rl}
Sample mean: & $\bar{X}_{n}=n^{-1}\sum_{i=1}^{n}X_{i}$, \\ 
Sample variance: & $S_{n}^{2}=(n-1)^{-1}\sum_{i=1}^{n}(X_{i}-\bar{X}_{n})^{2}$.
\end{tabular}
\end{equation*}

According to \cite{Robbins (1959)}, the MRPE for $\mu$ under the squared-error loss plus linear cost of sampling can be formulated as follows. Define the loss function by
\begin{equation}\label{3.1}
L_{n} \equiv L_{n}\left(\mu,\bar{X}_{n}\right)=A\left(\bar{X}_{n}-\mu\right)^{2}+cn,
\end{equation}
where $A(>0)$ is a known weight function and $c(>0)$ is the known unit cost of each observation. Associated with the loss function in \eqref{3.1}, we have the following risk function:
\begin{equation*}\label{3.2}
R_{n} \equiv E_{\mu,\sigma}\left[L_{n}\left(\mu,\bar{X}_{n}\right)\right]=A\sigma^{2}n^{-1}+cn,
\end{equation*}
which is minimized at 
\begin{equation}\label{3.3}
n^{\ast} \equiv n^{\ast}(c)=\sigma\sqrt{A/c},
\end{equation}
with the resulting minimum risk 
\begin{equation}\label{3.4}
R_{n^{\ast}}=2cn^{\ast}.
\end{equation}

One should note that there exists no fixed-sample-size procedure that achieves the exact minimum risk due to the fact that $\sigma$ is unknown. A fundamental solution is due to the purely sequential MRPE methodology in the light of \cite{Robbins (1959)}, \cite{Starr (1966)} and \cite{Starr and Woodroofe (1969)}, briefly introduced below.

Since the population standard deviation $\sigma$ remains unknown, it is essential for us to estimate it customarily using the sample standard deviation $S_{n}$, and update its value at every stage. One may start with $m$ pilot observations, $X_{1},...,X_{m},m\ge 2,$ and then sample one additional observation at-a-time as needed until the following stopping rule is satisfied:
\begin{equation}\label{3.5}
\text{Methodology} \ \mathcal{P}_{0}: N_{\mathcal{P}_{0}}\equiv N_{\mathcal{P}_{0}}(c)=\inf \left\{n\ge m:n\ge S_{n}\sqrt{A/c}\right\}.
\end{equation}
It is clear that $P_{\mu,\sigma}\{N_{\mathcal{P}_{0}}<\infty \}=1$ and $N_{\mathcal{P}_{0}}\uparrow \infty $ w.p.1 as $c\downarrow 0$. Upon termination with the accrued data $\{N_{\mathcal{P}_{0}},X_{1},...,X_{m},...,X_{N_{\mathcal{P}_{0}}}\}$, we estimate the unknown normal mean $\mu $ with $\bar{X}_{N_{\mathcal{P}_{0}}}\equiv N_{\mathcal{P}_{0}}^{-1}\sum_{i=1}^{N_{\mathcal{P}_{0}}}X_{i}$. The achieved risk is then given by
\begin{equation}\label{3.6}
R_{N_{\mathcal{P}_{0}}}(c)\equiv E_{\mu,\sigma}\left[L_{N_{\mathcal{P}_{0}}}\left(\mu,\bar{X}_{N_{\mathcal{P}_{0}}}\right)\right]=AE_{\mu,\sigma}\left[\left(\bar{X}_{N_{\mathcal{P}_{0}}}-\mu\right)^{2}\right]+cE_{\mu,\sigma}[N_{\mathcal{P}_{0}}].
\end{equation}

To measure the closeness between the achieved risk in \eqref{3.6} and the minimum risk in \eqref{3.4}, \cite{Robbins (1959)} and \cite{Starr (1966)} respectively constructed the following two crucial notions, namely, the \textit{risk efficiency} and \textit{regret}:
\begin{equation*}\label{3.7}
\begin{tabular}{rl}
(i) & Risk Efficiency: $\xi_{\mathcal{P}_{0}}(c)\equiv R_{N_{\mathcal{P}_{0}}}(c)/R_{n^{\ast }}(c)=\frac{1}{2}E_{\mu ,\sigma }[N_{\mathcal{P}_{0}}/n^{\ast }]+\frac{1}{2}E_{\mu ,\sigma }[n^{\ast }/N_{\mathcal{P}_{0}}]$;
\\
(ii) & Regret: $\omega_{\mathcal{P}_{0}}(c)\equiv R_{N_{\mathcal{P}_{0}}}(c)-R_{n^{\ast }}(c)=cE_{\mu,\sigma}\left[N_{\mathcal{P}_{0}}^{-1}(N_{\mathcal{P}_{0}}-n^{\ast})^{2}\right].$
\end{tabular}
\end{equation*}

Alternatively, the stopping rule \eqref{3.5} can be rewritten in a way that we presented \eqref{1.1}. By using the Helmert transformation, we express $N_{\mathcal{P}_{0}}=N_{\mathcal{P}_{0}}^{\prime}+1$ w.p.1, where the new stopping time $N_{\mathcal{P}_{0}}^{\prime}$ is defined as follows:
\begin{equation}\label{3.8}
N_{\mathcal{P}_{0}}^{\prime}=\inf\left\{n\ge m-1:n^{-1}\sum_{i=1}^{n}W_{i}\le (n/n^{\ast})^{2}(1+2n^{-1}+n^{-2})\right\},
\end{equation}
with $\delta=2$, $l_{0}=2$, and $W_{1},W_{2},...$ being i.i.d. $\chi_{1}^{2}$ random variables such that $\theta=1,\tau^{2}=1$ and $\alpha=1/2 $. From \cite{Robbins (1959)}, \cite{Starr (1966)} and \cite{Woodroofe (1977)}, we conclude the following theorem to address the asymptotic first-order and second-order properties that the purely sequential MRPE methodology $\mathcal{P}_{0}$ enjoys. One may refer to \cite{Mukhopadhyay and de Silva (2009)} for more details.

\begin{theorem}\label{Theorem 3.1}
For the purely sequential MRPE methodology $\mathcal{P}_{0}$ given in \eqref{3.5}, for all fixed $\mu,\sigma,m$ and $A$, we have: as $c\rightarrow 0$,
\begin{enumerate}
	\item[(i)] Asymptotic First-Order Efficiency: $E_{\mu,\sigma}\left[N_{\mathcal{P}_{0}}/n^{\ast}\right]\rightarrow 1$ if $m\geq 2$;
	\item[(ii)] Asymptotic Second-Order Efficiency: $E_{\mu,\sigma}\left[N_{\mathcal{P}_{0}}-n^{\ast}\right]=\eta_1(1)+o(1)$ if $m\geq 3$, where $\eta_1(1)=-\frac{1}{2}\sum_{n=1}^{\infty}n^{-1}E\left[\left\{\chi_{n}^{2}-3n\right\}^+\right]$;
	\item[(iii)] Asymptotic First-Order Risk Efficiency: $\xi_{\mathcal{P}_{0}}\left( c\right) \rightarrow 1$ if $m\geq 3$;
	\item[(iv)] Asymptotic Second-Order Risk Efficiency: $\omega_{\mathcal{P}_{0}}(c)=\frac{1}{2}c+o(c)$ if $m\geq 4$.
\end{enumerate}
\end{theorem}

\subsection{The general sequential MRPE methodology}\label{Sub 3.1}

Following \eqref{2.3}, we propose a broader and more general sequential MRPE methodology $\mathcal{P}(\rho,k)$:
\begin{equation}\label{3.10}
\begin{split}
T_{\mathcal{P}(\rho,k)} &\equiv T_{\mathcal{P}(\rho,k)}(c)=\inf \left\{n\geq0:m+kn\ge \rho S_{m+kn}\sqrt{A/c}\right\},\\
N_{\mathcal{P}(\rho,k)} &\equiv N_{\mathcal{P}(\rho,k)}(c)=\left\lfloor\rho^{-1}(m+kT_{\mathcal{P}(\rho ,k)})\right\rfloor +1.
\end{split}
\end{equation}
Here, $0<\rho \le 1$ is a prefixed proportion, $k\ge 1$ is a prefixed positive integer, $m\ge 2$ again indicates a pilot sample size but picked such that $m-1\equiv 0{\pmod k}$, and $\left\lfloor u\right\rfloor$ continues to denote the largest integer that is strictly smaller than $u$. Denote that $m-1=m_{0}k$ for some integer $m_{0}\ge 1$, we further assume that the following limit operations hold:
\begin{equation}\label{3.11}
m_{0}\rightarrow \infty, m=m_{0}k+1 \rightarrow \infty, c\equiv c(m)=O(m^{-2r}), n^{\ast}=O(m^{r}), \text{ and }\lim\sup\frac{m}{n^{\ast}}<\rho,
\end{equation}
where $r>1$ is a fixed constant. The new methodology $\mathcal{P}(\rho,k)$ is implemented as follows.

Starting with $m(=m_{0}k+1)$ pilot observations, $X_{1},...,X_{m},$ we sample $k$ observations at-a-time as needed and determine $T_{\mathcal{P}(\rho,k)}$, which indicates the number of sequential sampling operations according to the stopping rule \eqref{3.10}. Next, we continue to sample $(N_{\mathcal{P}(\rho,k)}-m-kT_{\mathcal{P}(\rho,k)})$ additional observations all in one batch. Upon termination, based on the fully gathered data
\begin{equation*}
\left\{T_{\mathcal{P}(\rho,k)},N_{\mathcal{P}(\rho,k)},X_{1},...,X_{m},...,X_{m+kT_{\mathcal{P}(\rho,k)}},...,X_{N_{\mathcal{P}(\rho,k)}}\right\},
\end{equation*}
we construct the minimum risk point estimator $\bar{X}_{N_{\mathcal{P}(\rho,k)}}=N_{\mathcal{P}(\rho,k)}^{-1}\sum_{i=1}^{N_{\mathcal{P}(\rho,k)}}X_{i}$ for $\mu$, and derive
\begin{equation*}\label{3.12}
\begin{tabular}{rl}
(i) & Risk Efficiency: $\xi_{\mathcal{P}(\rho,k)}(c)\equiv R_{N_{\mathcal{P}(\rho ,k)}}(c)/R_{n^{\ast}}(c)=\frac{1}{2}E_{\mu,\sigma }[N_{\mathcal{P}(\rho,k)}/n^{\ast}]+\frac{1}{2}E_{\mu,\sigma}[n^{\ast}/N_{\mathcal{P}(\rho,k)}]$;\\ 
(ii) & Regret: $\omega_{\mathcal{P}(\rho,k)}(c)\equiv R_{N_{\mathcal{P}(\rho,k)}}(c)-R_{n^{\ast }}(c)=cE_{\mu,\sigma }\left[ N_{\mathcal{P}(\rho,k)}^{-1}(N_{\mathcal{P}(\rho,k)}-n^{\ast})^{2}\right] $.
\end{tabular}
\end{equation*}

Obviously, $P_{\mu,\sigma}(N_{\mathcal{P}(\rho,k)}<\infty)=1$ and $N_{\mathcal{P}(\rho,k)}\uparrow \infty $ w.p.1 as $c\downarrow 0$. If both $\rho$ and $k$ are chosen to be 1, then the sequential MRPE methodology $\mathcal{P}(1,1)$ will be the ordinary purely sequential MRPE methodology $\mathcal{P}_{0}$ as per \eqref{3.5}. That is, $\mathcal{P}(1,1)\equiv \mathcal{P}_{0}$.

Along the line of \eqref{3.8}, we can similarly express the stopping time $T_{\mathcal{P}(\rho,k)}$ from \eqref{3.10} in the general form provided in \eqref{2.1}. Define $T_{\mathcal{P}(\rho,k)}=T_{\mathcal{P}(\rho,k)}^{\prime}-m_{0}$ w.p.1. Then $T_{\mathcal{P}(\rho,k)}^{\prime}$ is a new stopping time which can be rewritten as
\begin{equation}\label{3.13}
T_{\mathcal{P}(\rho,k)}^{\prime}=\inf\left\{n\ge m_{0}:(kn)^{-1}\sum_{i=1}^{n}U_{i}\le [kn/(\rho n^{\ast})]^{2}\left(1+2(kn)^{-1}+(kn)^{-2}\right) \right\},
\end{equation}
where $\delta=2,l_{0}=2$, and $U_{i}=\sum_{j=(i-1)k+1}^{ik}W_{j},i=1,2,...$ with $W_{1},W_{2},...$ being i.i.d. $\chi_{1}^{2}$ random variables such that $\theta=1,\tau^{2}=2$ and $\alpha=1/2$. Therefore, $U_{1},U_{2},...$ are i.i.d. $\chi_{k}^{2}$ random variables. Now we state a number of asymptotic first-order and second-order properties of the improved accelerated sequential MRPE methodology $\mathcal{P}(\rho,k)$, summarized in the following theorem.
\begin{theorem}\label{Theorem 3.2}
For the general sequential MRPE methodology $\mathcal{P}(\rho,k)$ given in \eqref{3.10}, for all fixed $\mu,\sigma,A, $k$ and 0<\rho<1$, under the limit operations \eqref{3.11}:
\begin{enumerate}
	\item[(i)] Asymptotic First-Order Efficiency: $E_{\mu,\sigma}\left[ N_{\mathcal{P}(\rho,k)}/n^{\ast}\right]\rightarrow 1$;
	\item[(ii)] Asymptotic Second-Order Efficiency: $\rho^{-1}\eta_1(k)+o(1)\leq E_{\mu,\sigma}\left[ N_{\mathcal{P}(\rho,k)}-n^{\ast}\right] \le \rho^{-1}\eta_1(k)+1+o(1)$, where $\eta_1(k)=\frac{k-1}{2}-\frac{1}{2}\sum_{n=1}^{\infty}n^{-1}E\left[ \left\{\chi_{kn}^{2}-3kn\right\}^{+} \right]$;
	\item[(iii)] Asymptotic First-Order Risk Efficiency: $\xi_{\mathcal{P}(\rho,k)}\left( c\right) \rightarrow 1$;
	\item[(iv)] Asymptotic Second-Order Risk Efficiency: $\omega_{\mathcal{P}(\rho,k)}(c)=\frac{1}{2}\rho^{-1}c+o(c)$.
\end{enumerate}
\end{theorem}

Again, when $\rho=1$, then we have the exact expression $E_{\mu,\sigma}\left[ N_{\mathcal{P}(1,k)}-n^{\ast}\right]=\eta_1(k)+o(1)$ instead of the inequality in Theorem \ref{Theorem 3.2} (ii). The number of sampling operations for the general sequential MRPE methodology $\mathcal{P}(\rho,k)$ is  
\begin{equation}\label{3.15}
\varphi_{\mathcal{P}(\rho,k)}=T_{\mathcal{P}(\rho,k)}+1+I(\rho<1),
\end{equation}
and
\begin{equation}\label{3.16}
E_{\mu,\sigma}[\varphi_{\mathcal{P}(\rho,k)}]=k^{-1}[\rho n^{\ast}-m+\eta_1(k)]+1+I(\rho <1)+o(1).
\end{equation}

For any integer $k\ge 1$, $\eta_1(k)=\frac{k-1}{2}-\frac{1}{2}\Sigma_{n=1}^{\infty}n^{-1}E\left[ \left\{\chi_{kn}^{2}-3kn\right\}^{+}\right]$ is computable. In order to obtain numerical approximations, we wrote out our own R codes and provided the values as in Table \ref{Table 1}. In the spirit of \citet[Table 3.8.1]{Mukhopadhyay and Solanky (1994)}, any term smaller than $10^{-15}$ in magnitude was excluded in the infinite sum with regard to $\eta_1(k)$. Intuitively, the infinite sum and $k^{-1}\eta_1(k)$ converge to zero and one half as $k\rightarrow \infty $, respectively. However, by looking at the columns of $\eta_1(k)$ and $k^{-1}\eta_1(k)$, one can see that the infinite sum converges very fast, while $k^{-1}\eta_1(k)$ converges at a rather slow rate.

\begin{table}[h!] 
\footnotesize
\captionsetup{font=footnotesize}
\caption{$\eta_1(k)$ approximations in Theorem \ref{Theorem 3.2} (ii)}
\label{Table 1}\par
\centerline{\tabcolsep=3truept
\begin{tabular}{ccccccc}
\hline
$k$ & $\eta_1(k)$ & $k^{-1}\eta_1(k)$ &  ~~~  & $k$ & $\eta_1(k)$ & $
k^{-1}\eta_1(k)$\\ 
\hline
$1$ & \multicolumn{1}{r}{$-0.1165$} & \multicolumn{1}{r}{$-0.1165$} &  & $11$ & $4.9993$ & $0.4545$\\ 
$2$ & \multicolumn{1}{r}{$0.4367$} & \multicolumn{1}{r}{$0.2183$} &  & $12$
& $5.4996$ & $0.4583$\\ 
$3$ & \multicolumn{1}{r}{$0.9636$} & \multicolumn{1}{r}{$0.3212$} &  & $13$
& $5.9997$ & $0.4615$\\ 
$4$ & \multicolumn{1}{r}{$1.4785$} & \multicolumn{1}{r}{$0.3696$} &  & $14$
& $6.4998$ & $0.4643$\\ 
$5$ & \multicolumn{1}{r}{$1.9872$} & \multicolumn{1}{r}{$0.3974$} &  & $15$
& $6.9999$ & $0.4667$\\ 
$6$ & \multicolumn{1}{r}{$2.4922$} & \multicolumn{1}{r}{$0.4154$} &  & $16$
& $7.4999$ & $0.4687$\\ 
$7$ & \multicolumn{1}{r}{$2.9952$} & \multicolumn{1}{r}{$0.4279$} &  & $17$
& $8.0000$ & $0.4706$\\ 
$8$ & \multicolumn{1}{r}{$3.4971$} & \multicolumn{1}{r}{$0.4371$} &  & $18$
& $8.5000$ & $0.4722$\\ 
$9$ & \multicolumn{1}{r}{$3.9982$} & \multicolumn{1}{r}{$0.4442$} &  & $19$
& $9.0000$ & $0.4737$\\ 
$10$ & \multicolumn{1}{r}{$4.4989$} & \multicolumn{1}{r}{$0.4499$} &  & $
20$ & $9.5000$ & $0.4750$\\ 
\hline
\end{tabular}}
\end{table}

\subsection{Simulated performances}\label{Sub 3.2}

To investigate the appealing properties of the general sequential MRPE methodology $\mathcal{P}(\rho,k)$, and illustrate how it saves sampling operations with $0<\rho <1$ and/or $k\geq 2$, we conducted extensive sets of simulations under the normal case in the spirit of \cite{Mukhopadhyay and Hu (2017)}. To be specific, we generated pseudo-random samples from a $N(5,2^{2})$ population. While fixing the weight function $A=100$, the pilot sample size $m=21$, we selected a wide range of values of $c$, the unit cost of sampling, including $0.04,0.01$ and $0.0025$ so that the optimal fixed sample size $n^{\ast}$ turned out to be $100,200$ and $400$ accordingly. We also considered various combinations of $\rho=(1,0.8,0.5)$ and $k=(1,2,5)$ to compare the number of sampling operations under different possible scenarios. The findings are summarized in Table \ref{Table 2}. For each methodology $\mathcal{P}(\rho,k)$, we computed the average total final sample size $\bar{n}$ with the associated standard error $s(\bar{n})$, the difference between $\bar{n}$ and $n^{\ast}$ to be compared with the second-order efficiency term in Theorem \ref{Theorem 3.2} (ii), the estimated risk efficiency $\widehat{\xi}$ to be compared with 1, the estimated regret in terms of unit cost $\widehat{\omega}/c$ to be compared with $\frac{1}{2}\rho^{-1}$ from Theorem \ref{Theorem 3.2} (iv),  as well as the average number of sampling operations $\bar{\varphi}$ to be compared with the expected number of sampling operations $E(\varphi)$ from \eqref{3.16}.

It is clear that across the board, $\bar{n}-n^*$ is close to the second-order approximation $\rho^{-1}\eta_1(k)$, $\hat{\xi}$ is close to 1, and $\hat{\omega}/c$ is close to the coefficient $\frac{1}{2}\rho^{-1}$. These empirically verify Theorem \ref{Theorem 3.2}. Focusing on the last two columns, we can also easily find that the average number of sampling operations needed, $\bar{\varphi}$, is almost the same with the theoretical value $E(\varphi)$, and the $k$-at-a-time improved accelerated sequential MRPE procedure $\mathcal{P}(\rho,k)$ reduces approximately $100(1-k^{-1}\rho)\%$ sampling operations in contrast to the Anscombe-Chow-Robbins purely sequential procedure $\mathcal{P}(1,1)$. For example, when $n^*=400$, $\mathcal{P}(1,1)$ requires around 380 sampling operations on average, while $\mathcal{P}(0.8,5)$ requires 62. So about $100(1-62/380)\%=83.7\%$ sampling operations are saved, which is close to $100(1-0.8/5)\%=84\%$.  

\begin{table}[t!]
\footnotesize
\captionsetup{font=footnotesize}
\caption{Simulations from $N(5,2^{2})$ with $A=100$ and $m=21$ under $10,000$ runs implementing $\mathcal{P}(\rho,k)$ from \eqref{3.10}}
\label{Table 2}\par
\vskip .2cm
\centerline{\tabcolsep=3truept\begin{tabular}{cccccccccccc} \hline 
$n^{\ast}$ & $c$ & $\mathcal{P}(\rho,k)$ & $\bar{n}$ & $s\left(\bar{n}\right) $ & $\bar{n}-n^{\ast}$ & $\rho^{-1}\eta_1(k)$ & $\widehat{\xi}$ & $\frac{1}{2}\rho^{-1}$ & $\widehat{\omega}/c$ & $\bar{\varphi}$ & $E(\varphi)$ \\ \hline
\multicolumn{1}{r}{$100$} & \multicolumn{1}{r}{$0.04$} & $\mathcal{P}(1,1)$
& \multicolumn{1}{r}{${99.8528}$} & \multicolumn{1}{r}{$0.07182$}
& \multicolumn{1}{r}{$-0.1472$} & \multicolumn{1}{r}{$-0.1165$} & 
\multicolumn{1}{r}{$0.99224$} & \multicolumn{1}{r}{$0.5$} & $0.53304$ & \multicolumn{1}{r}{79.853} & \multicolumn{1}{r}{79.883} \\ 
\multicolumn{1}{r}{} & \multicolumn{1}{r}{} & $\mathcal{P}(1,2)$ & 
\multicolumn{1}{r}{$100.4296$} & \multicolumn{1}{r}{$0.07231$} & 
\multicolumn{1}{r}{$0.4296$} & \multicolumn{1}{r}{$0.4367$} & 
\multicolumn{1}{r}{$0.99239$} & \multicolumn{1}{r}{$0.5$} & $0.53390$  & \multicolumn{1}{r}{40.715} & \multicolumn{1}{r}{40.718} \\ 
&  & $\mathcal{P}(1,5)$ & \multicolumn{1}{r}{$102.0110$} & $0.07340$ & 
\multicolumn{1}{r}{$2.0110$} & \multicolumn{1}{r}{$1.9872$} & $0.99314$ & 
\multicolumn{1}{r}{$0.5$} & $0.56211$  & \multicolumn{1}{r}{17.202} & \multicolumn{1}{r}{17.197} \\ 
&  & $\mathcal{P}(0.8,1)$ & \multicolumn{1}{r}{$100.1738$} & $0.08124$ & 
\multicolumn{1}{r}{$0.1738$} & \multicolumn{1}{r}{$-0.1456$} & $0.99291$ & 
\multicolumn{1}{r}{$0.625$} & $0.68228$  & \multicolumn{1}{r}{60.836} & \multicolumn{1}{r}{60.883} \\ 
&  & $\mathcal{P}(0.8,2)$ & \multicolumn{1}{r}{$101.0466$} & $0.08049$ & 
\multicolumn{1}{r}{$1.0466$} & \multicolumn{1}{r}{$0.5459$} & $0.99339$ & 
\multicolumn{1}{r}{$0.625$} & $0.66320$  & \multicolumn{1}{r}{$31.718$} & \multicolumn{1}{r}{$31.718$} \\ 
&  & $\mathcal{P}(0.8,5)$ & \multicolumn{1}{r}{$102.8798$} & $0.08368$ & 
\multicolumn{1}{r}{$2.8798$} & \multicolumn{1}{r}{$2.4840$} & $0.99409$ & 
\multicolumn{1}{r}{$0.625$} & $0.74873$  & \multicolumn{1}{r}{14.195} & \multicolumn{1}{r}{14.197} \\ 
&  & $\mathcal{P}(0.5,1)$ & \multicolumn{1}{r}{$99.8002$} & $0.10324$ & 
\multicolumn{1}{r}{$-0.1998$} & \multicolumn{1}{r}{$-0.2330$} & $0.99513$ & 
\multicolumn{1}{r}{$1$} & $1.14253$  & \multicolumn{1}{r}{30.900} & \multicolumn{1}{r}{30.884} \\ 
\multicolumn{1}{r}{} & \multicolumn{1}{r}{} & $\mathcal{P}(0.5,2)$ & 
\multicolumn{1}{r}{$100.8176$} & \multicolumn{1}{r}{$0.10405$} & 
\multicolumn{1}{r}{$0.8176$} & \multicolumn{1}{r}{$0.8734$} & 
\multicolumn{1}{r}{$0.99524$} & \multicolumn{1}{r}{$1$} & $1.12474$  & \multicolumn{1}{r}{16.704} & \multicolumn{1}{r}{16.718} \\ 
&  & $\mathcal{P}(0.5,5)$ & \multicolumn{1}{r}{$\underset{}{103.8990}$} & $
0.10664$ & \multicolumn{1}{r}{$3.8990$} & \multicolumn{1}{r}{$3.9744$} & $
0.99697$ & \multicolumn{1}{r}{$1$} & $1.21588$  & \multicolumn{1}{r}{8.190} & \multicolumn{1}{r}{8.197} \\ \hline
\multicolumn{1}{r}{$200$} & \multicolumn{1}{r}{$0.01$} & $\mathcal{P}(1,1)$
& \multicolumn{1}{r}{$\overset{}{199.9278}$} & \multicolumn{1}{r}{$0.10018$}
& \multicolumn{1}{r}{$-0.0722$} & \multicolumn{1}{r}{$-0.1165$} & 
\multicolumn{1}{r}{$0.99652$} & \multicolumn{1}{r}{$0.5$} & $0.50880$  & \multicolumn{1}{r}{179.928} & \multicolumn{1}{r}{179.884} \\ 
\multicolumn{1}{r}{} & \multicolumn{1}{r}{} & $\mathcal{P}(1,2)$ & 
\multicolumn{1}{r}{$200.4926$} & \multicolumn{1}{r}{$0.10056$} & 
\multicolumn{1}{r}{$0.4926$} & \multicolumn{1}{r}{$0.4367$} & 
\multicolumn{1}{r}{$0.99646$} & \multicolumn{1}{r}{$0.5$} & $0.50928$  & \multicolumn{1}{r}{$90.746$} & \multicolumn{1}{r}{$90.718$} \\ 
&  & $\mathcal{P}(1,5)$ & $202.0495$ & $0.10108$ & \multicolumn{1}{r}{$2.0495
$} & \multicolumn{1}{r}{$1.9872$} & $0.99664$ & \multicolumn{1}{r}{$0.5$} & $
0.52308$  & \multicolumn{1}{r}{$37.210$} & \multicolumn{1}{r}{$37.197$} \\ 
&  & $\mathcal{P}(0.8,1)$ & $200.3638$ & $0.11258$ & \multicolumn{1}{r}{$
0.3638$} & \multicolumn{1}{r}{$-0.1456$} & $0.99680$ & \multicolumn{1}{r}{$
0.625$} & $0.64117$  & \multicolumn{1}{r}{$140.991$} & \multicolumn{1}{r}{$140.884$} \\ 
&  & $\mathcal{P}(0.8,2)$ & $201.0546$ & $0.11208$ & \multicolumn{1}{r}{$
1.0546$} & \multicolumn{1}{r}{$0.5459$} & $0.99686$ & \multicolumn{1}{r}{$
0.625$} & $0.63342$  & \multicolumn{1}{r}{$71.721$} & \multicolumn{1}{r}{$71.718$} \\ 
&  & $\mathcal{P}(0.8,5)$ & $202.8111$ & $0.11321$ & \multicolumn{1}{r}{$
2.8111$} & \multicolumn{1}{r}{$2.4840$} & $0.99709$ & \multicolumn{1}{r}{$
0.625$} & $0.66353$  & \multicolumn{1}{r}{$30.189$} & \multicolumn{1}{r}{$30.197$} \\ 
&  & $\mathcal{P}(0.5,1)$ & $199.7988$ & $0.14367$ & \multicolumn{1}{r}{$
-0.2012$} & \multicolumn{1}{r}{$-0.2330$} & $0.99776$ & \multicolumn{1}{r}{$1
$} & $1.06174$  & \multicolumn{1}{r}{$80.899$} & \multicolumn{1}{r}{$80.884$} \\ 
\multicolumn{1}{r}{} & \multicolumn{1}{r}{} & $\mathcal{P}(0.5,2)$ & 
\multicolumn{1}{r}{$200.9172$} & \multicolumn{1}{r}{$0.14357$} & 
\multicolumn{1}{r}{$0.9172$} & \multicolumn{1}{r}{$0.8734$} & 
\multicolumn{1}{r}{$0.99778$} & \multicolumn{1}{r}{$1$} & $1.04809$  & \multicolumn{1}{r}{$41.729$} & \multicolumn{1}{r}{$41.718$} \\ 
\multicolumn{1}{r}{} & \multicolumn{1}{r}{} & $\mathcal{P}(0.5,5)$ & 
\multicolumn{1}{r}{$\underset{}{203.9490}$} & \multicolumn{1}{r}{$0.14531$}
& \multicolumn{1}{r}{$3.9490$} & \multicolumn{1}{r}{$3.9744$} & 
\multicolumn{1}{r}{$0.99812$} & \multicolumn{1}{r}{$1$} & $1.09737$  & \multicolumn{1}{r}{$18.195$} & \multicolumn{1}{r}{$18.197$} \\ \hline
\multicolumn{1}{r}{$400$} & \multicolumn{1}{r}{$0.0025$} & $\mathcal{P}(1,1)$
& \multicolumn{1}{r}{$\overset{}{399.8625}$} & \multicolumn{1}{r}{$0.14074$}
& \multicolumn{1}{r}{$-0.1375$} & \multicolumn{1}{r}{$-0.1165$} & 
\multicolumn{1}{r}{$0.99805$} & \multicolumn{1}{r}{$0.5$} & $0.49788$  & \multicolumn{1}{r}{$379.863$} & \multicolumn{1}{r}{$379.883$} \\ 
\multicolumn{1}{r}{} & \multicolumn{1}{r}{} & $\mathcal{P}(1,2)$ & 
\multicolumn{1}{r}{$400.4010$} & \multicolumn{1}{r}{$0.14155$} & 
\multicolumn{1}{r}{$0.4010$} & \multicolumn{1}{r}{$0.4367$} & 
\multicolumn{1}{r}{$0.99803$} & \multicolumn{1}{r}{$0.5$} & $0.50278$  & \multicolumn{1}{r}{$190.701$} & \multicolumn{1}{r}{$190.718$} \\ 
&  & $\mathcal{P}(1,5)$ & $401.9310$ & $0.14169$ & \multicolumn{1}{r}{$1.9310
$} & \multicolumn{1}{r}{$1.9872$} & $0.99806$ & \multicolumn{1}{r}{$0.5$} & $
0.50681$  & \multicolumn{1}{r}{$77.186$} & \multicolumn{1}{r}{$77.197$} \\ 
&  & $\mathcal{P}(0.8,1)$ & $400.1549$ & $0.15943$ & \multicolumn{1}{r}{$
0.1549$} & \multicolumn{1}{r}{$-0.1456$} & $0.99815$ & \multicolumn{1}{r}{$
0.625$} & $0.63931$  & \multicolumn{1}{r}{$300.824$} & \multicolumn{1}{r}{$300.883$} \\ 
&  & $\mathcal{P}(0.8,2)$ & $401.0467$ & $0.15757$ & \multicolumn{1}{r}{$
1.0467$} & \multicolumn{1}{r}{$0.5459$} & $0.99821$ & \multicolumn{1}{r}{$
0.625$} & $0.62297$  & \multicolumn{1}{r}{$151.718$} & \multicolumn{1}{r}{$151.718$} \\ 
&  & $\mathcal{P}(0.8,5)$ & $402.8376$ & $0.15905$ & \multicolumn{1}{r}{$
2.8376$} & \multicolumn{1}{r}{$2.4840$} & $0.99824$ & \multicolumn{1}{r}{$
0.625$} & $0.64402$  & \multicolumn{1}{r}{$62.194$} & \multicolumn{1}{r}{$62.197$} \\ 
&  & $\mathcal{P}(0.5,1)$ & $399.5954$ & $0.20307$ & \multicolumn{1}{r}{$
-0.4046$} & \multicolumn{1}{r}{$-0.2330$} & $0.99865$ & \multicolumn{1}{r}{$1
$} & $1.04584$  & \multicolumn{1}{r}{$180.798$} & \multicolumn{1}{r}{$180.883$} \\ 
& \multicolumn{1}{r}{} & $\mathcal{P}(0.5,2)$ & \multicolumn{1}{r}{$400.9764$
} & \multicolumn{1}{r}{$0.20288$} & \multicolumn{1}{r}{$0.9764$} & 
\multicolumn{1}{r}{$0.8734$} & \multicolumn{1}{r}{$0.99870$} & 
\multicolumn{1}{r}{$1$} & $1.03536$  & \multicolumn{1}{r}{$91.744$} & \multicolumn{1}{r}{$97.718$} \\ 
& \multicolumn{1}{r}{} & $\mathcal{P}(0.5,5)$ & \multicolumn{1}{r}{$403.9610$
} & \multicolumn{1}{r}{$0.20171$} & \multicolumn{1}{r}{$3.9610$} & 
\multicolumn{1}{r}{$3.9744$} & \multicolumn{1}{r}{$0.99875$} & 
\multicolumn{1}{r}{$1$} & $1.03629$  & \multicolumn{1}{r}{$38.196$} & \multicolumn{1}{r}{$38.197$} \\ \hline
\end{tabular}}
\end{table}

\subsection{Real data analysis}\label{Sub 3.3}

Next, to illustrate the applicability of our newly developed general sequential MRPE methodology $\mathcal{P}(\rho,k)$, we proceeded to analyze a real-life dataset on hospital infection data from \cite{Kutner et al. (2005)}. This data is from 113 hospitals in the United States for the 1975-76 study period. Each line of the data set has an identification number and provides information on 11 other variables for a single hospital. One of the 12 variables is the infection risk, which records the average estimated probability of acquiring infection in hospital (in percent). With the cost of observations taken into consideration, it is of great interest to propose an MRPE for the infection risk.

We treated the real data set on the infection risk, which seemed to follow a normal distribution, confirmed via Shapiro-Wilk normality test with the associated $p$-value of $0.1339$. The simple descriptive statistics from the
whole data set of infection risk are summarized as follows.

\begin{equation*}
\begin{tabular}{cccccccc}
\hline
$n$ & $\bar{x}$ & $s$ & Min & $Q_{1}$ & Med & $Q_{3}$ & Max \\ 
\multicolumn{1}{l}{$113$} & \multicolumn{1}{l}{$4.355$} & \multicolumn{1}{l}{$1.341$} & \multicolumn{1}{l}{$1.300$} & \multicolumn{1}{l}{$3.700$} & 
\multicolumn{1}{l}{$4.400$} & \multicolumn{1}{l}{$5.200$} & 
\multicolumn{1}{l}{$7.800$} \\ \hline
\end{tabular}
\end{equation*}

For illustrative purposes, we treated this data set of infection risk with size 113 from \cite{Kutner et al. (2005)} as our population with both mean and variance assumed unknown. Then, we performed our general sequential MRPE methodologies to obtain minimum risk point estimators for the infection risk. To start, we first randomly picked $m=11$ observations as a pilot sample, based upon which we proceeded with sampling according to the methodologies $\mathcal{P}\left(\rho,k\right)$ with $A=100,c=0.04,$ $\rho=(1,0.8,0.5),$ $k=(1,2,5)$, respectively. We summarized the terminated sample sizes as well as the associated numbers of sampling operations under each setting in Table \ref{Table 3}, where $\mathcal{P}(\rho,k)$ denotes a certain sampling procedure with fixed values of $\rho$ and $k$, and $n_{\mathcal{P}(\rho,k)}$ and $\varphi_{\mathcal{P}(\rho,k)}$ indicate the respective terminated sample size and number of sampling operations performing $\mathcal{P}\left(\rho,k\right)$ accordingly.

From Table \ref{Table 3}, we can see our terminated sample size ranges from 54 to 77. Apparently enough, much fewer sampling operations are needed when we fix $k=2,5$, without increasing a significant number of observations. Also, we
need the least observations with least sampling operations when we use the methodology $\mathcal{P}\left(\rho,k\right) $ with $\rho=0.5$, comparing with larger $\rho=0.8$ or $1$, for a fixed $k$ value. The point estimates constructed from each sampling procedure were listed in the last column, and they were close to each other. Finally, one should reiterate that each row in Table \ref{Table 3} was obtained from one single run, but shows the practical applicability of our $k$-at-a-time improved accelerated sequential MRPE methodology $\mathcal{P}\left(\rho,k\right)$. We had indeed repeated similar implementations, but no obvious difference appeared. Consequently, we leave out many details for brevity.

\begin{table}[t!] 
\footnotesize
\captionsetup{font=footnotesize}
\caption{Terminated sample size associated with number of sampling operations using $\mathcal{P}(\rho ,k)$ as per \eqref{3.10} }
\label{Table 3}\par
\vskip .2cm
\centerline{\tabcolsep=3truept\begin{tabular}{cccc}
\hline
$\mathcal{P}(\rho ,k)$ & $n_{\mathcal{P}(\rho,k)}$ & $\varphi_{\mathcal{P}(\rho,k)}$&$\hat{\mu}$\\ 
\hline
$\mathcal{P}(1,1)$ & $70$ & $60$ & $4.4471$\\ 
$\mathcal{P}(1,2)$ & $73$ & $32$ &$4.4342$\\ 
$\mathcal{P}(1,5)$ & $76$ & $14$ &$4.4526$\\ 
$\mathcal{P}(0.8,1)$ & $72$ & $48$ &$4.4306$\\ 
$\mathcal{P}(0.8,2)$ & $72$ & $25$ &$4.4306$\\ 
$\mathcal{P}(0.8,5)$ & $77$ & $12$ &$4.4532$\\ 
$\mathcal{P}(0.5,1)$ & $54$ & $18$ &$4.4519$\\ 
$\mathcal{P}(0.5,2)$ & $54$ & $10$ &$4.4519$\\ 
$\mathcal{P}(0.5,5)$ & $62$ & $6$ &$4.4565$\\ \hline
\end{tabular}}
\end{table}


\setcounter{equation}{0}
\section{Bounded Variance Point Estimation for Negative Exponential Location}\label{Sect. 4}

For a fixed sample size, however large it is, the variance of an estimator can be larger than a prescribed level to an arbitrary extent. This problem has been addressed in \cite{Hu and Hong (2021)}, where the authors focused on estimating the pure premium in actuarial science. Here, our newly proposed general sequential sampling scheme $\mathcal{M}(\rho,k)$ can be implemented to guarantee that the variance of our estimator is close to all small predetermined levels. In this section, therefore, we include another illustration which is the bounded variance point estimation (BVPE) for the location parameter $\mu$ of a negative exponential distribution $NExp(\mu,\sigma)$ with the probability density function $$f(y;\mu,\sigma)=\frac{1}{\sigma}\exp\left\{-\frac{y-\mu}{\sigma}\right\}I(y>\mu),$$ where both $\mu$ and $\sigma$ remain unknown. Having recorded a random sample $Y_1,...,Y_n,n\ge2$, we denote $Y_{n:1}=\min\{Y_1,...,Y_n\}$, which is the maximum likelihood estimator (MLE) of $\mu$, and $V_n=(n-1)^{-1}\sum_{i=1}^{n}(Y_i-Y_{n:1})$, which is the uniformly minimum variance unbiased estimator (UMVUE) of $\sigma$. As a standard approach, we estimate $\mu$ using its MLE $Y_{n:1}$, which is a consistent estimator.

It is well known that (i) $n(Y_{n:1}-\mu)/\sigma \sim NExp(0,1)$; (ii) $2(n-1)V_n/\sigma \sim \chi^2_{2n-2}$; and (iii) $Y_{n:1}$ and $(V_2,...,V_n),n\ge2$ are independent. Hence, the variance of the proposed point estimator $Y_{n:1}$ is
\begin{equation}\label{4.2}
V_{\mu,\sigma}[Y_{n:1}] = \frac{\sigma^2}{n^2}.
\end{equation} 
Now, our goal is to make $V_{\mu,\sigma}[Y_{n:1}]$ fall below (or be close to) a predetermined level $b^2,b>0$ for all $0<\sigma<\infty$. Then, it is clear that we have $n \ge \sigma/b$. The optimal fixed sample size is therefore given by 
\begin{equation}\label{4.3}
n^* = \frac{\sigma}{b}.
\end{equation}
One may refer to \citet[p. 183]{Mukhopadhyay and de Silva (2009)} for more information.

Since $\sigma$ is unknown to us, we estimate it by updating its UMVUE $V_n$ at every stage as needed, and implement the following general sequential BVPE methodology $\mathcal{Q}(\rho,k)$: 
\begin{equation}\label{4.4}
\begin{split}
T_{\mathcal{Q}(\rho,k)} &\equiv T_{\mathcal{Q}(\rho,k)}(c)=\inf \left\{n\geq0:m+kn\ge \rho V_{m+kn}/b \right\},\\
N_{\mathcal{Q}(\rho,k)} &\equiv N_{\mathcal{Q}(\rho,k)}(c)=\left\lfloor\rho^{-1}(m+kT_{\mathcal{Q}(\rho,k)})\right\rfloor +1.
\end{split}
\end{equation}
Here, $0<\rho \le 1$ is a prefixed proportion, $k\ge 1$ is a prefixed positive integer, and the pilot sample size $m=m_0k+1$ for some $m_0$. We further assume that the following limit operations hold:
\begin{equation}\label{4.5}
m_{0}\rightarrow \infty, m=m_{0}k+1 \rightarrow \infty, b\equiv b(m)=O(m^{-r}), n^{\ast}=O(m^{r}), \text{ and }\lim\sup\frac{m}{n^{\ast}}<\rho,
\end{equation}
where $r>1$ is a fixed constant. The methodology $\mathcal{Q}(\rho,k)$ is conducted analogously with the methodology $\mathcal{P}(\rho,k)$ introduced in Section \ref{Sect. 3}.

Again, it is clear that $P_{\mu,\sigma}(N_{\mathcal{Q}(\rho,k)}<\infty)=1$ and $N_{\mathcal{Q}(\rho,k)}\uparrow \infty $ w.p.1 as $b\downarrow 0$. Upon termination with the fully gathered data $$\left\{ T_{\mathcal{Q}(\rho,k)},N_{\mathcal{Q}(\rho,k)},Y_{1},...,Y_{m},...,Y_{m+kT_{\mathcal{Q}(\rho,k)}},...,Y_{N_{\mathcal{Q}(\rho,k)}} \right\},$$
we construct the bounded variance point estimator $Y_{N_{\mathcal{Q}(\rho,k)}:1}=\min\{Y_1,...,Y_{N_\mathcal{Q}(\rho,k)}\}$ for $\mu$.

Along the line of \eqref{3.13}, we define $T_{\mathcal{Q}(\rho,k)}=T_{\mathcal{Q}(\rho,k)}^{\prime}-m_{0}$ w.p.1. Then $T_{\mathcal{Q}(\rho,k)}^{\prime}$ is a new stopping time which can be rewritten as
\begin{equation}\label{4.6}
T_{\mathcal{Q}(\rho,k)}^{\prime}=\inf\left\{n\ge m_{0}:(kn)^{-1}\sum_{i=1}^{n}U_{i}\le 2[kn/(\rho n^{\ast})]\left(1+(kn)^{-1}\right) \right\},
\end{equation}
where $\delta=1,l_{0}=1$, and $U_{i}=\sum_{j=(i-1)k+1}^{ik}W_{j},i=1,2,...$ with $W_{1},W_{2},...$ being i.i.d. $\chi_{2}^{2}$ random variables such that $\theta=2,\tau^{2}=4$ and $\alpha=1$. Therefore, $U_{1},U_{2},...$ are i.i.d. $\chi_{2k}^{2}$ random variables. Now we state a number of asymptotic first-order and second-order properties of the general sequential BVPE methodology $\mathcal{Q}(\rho,k)$, summarized in the following theorem.
\begin{theorem}\label{Theorem 4.1}
For the general sequential BVPE methodology $\mathcal{Q}(\rho,k)$ given in \eqref{4.4}, for all fixed $\mu,\sigma,k$ and $0<\rho<1$, under the limit operations \eqref{4.5}:
\begin{enumerate}
	\item[(i)] Asymptotic First-Order Efficiency: $E_{\mu,\sigma}\left[ N_{\mathcal{Q}(\rho,k)}/n^{\ast}\right]\rightarrow 1$;
	\item[(ii)] Asymptotic Second-Order Efficiency: $\rho^{-1}\eta_2(k)+o(1)\leq E_{\mu,\sigma}\left[ N_{\mathcal{Q}(\rho,k)}-n^{\ast}\right] \le \rho^{-1}\eta_2(k)+1+o(1)$, where $\eta_2(k)=\frac{k-1}{2}-\frac{1}{2}\sum_{n=1}^{\infty}n^{-1}E\left[ \left\{\chi_{2kn}^{2}-4kn\right\}^{+} \right]$;
	\item[(iii)] Asymptotic Variance: $V_{\mu,\sigma}[Y_{N_{\mathcal{Q}(\rho,k)}:1}]=b^2+o(b^2)$.
\end{enumerate}
\end{theorem}

When $\rho=1$, then we have the exact expression $E_{\mu,\sigma}\left[ N_{\mathcal{Q}(1,k)}-n^*\right]=\eta_2(k)+o(1)$ instead of the inequality in Theorem \ref{Theorem 4.1} (ii). The number of sampling operations for the general sequential BVPE methodology $\mathcal{Q}(\rho,k)$ is  
\begin{equation}\label{4.7}
\varphi_{\mathcal{Q}(\rho,k)}=T_{\mathcal{Q}(\rho,k)}+1+I(\rho<1),
\end{equation}
and
\begin{equation}\label{4.8}
E_{\mu,\sigma}[\varphi_{\mathcal{Q}(\rho,k)}]=k^{-1}[\rho n^*-m+\eta_2(k)]+1+I(\rho<1)+o(1).
\end{equation}

For any integer $k\ge 1$, $\eta_2(k)=\frac{k-1}{2}-\frac{1}{2}\sum_{n=1}^{\infty}n^{-1}E\left[ \left\{\chi_{2kn}^{2}-4kn\right\}^{+}\right]$ is also computable. Table \ref{Table 4} provides some numerical approximations in the same fashion of Table \ref{Table 1}.

\begin{table}[h!] 
\footnotesize
\captionsetup{font=footnotesize}
\caption{$\eta_2(k)$ approximations in Theorem \ref{Theorem 4.1} (ii)}
\label{Table 4}\par
\centerline{\tabcolsep=3truept
\begin{tabular}{ccccccc}
\hline
$k$ & $\eta_2(k)$ & $k^{-1}\eta_2(k)$ &  ~~~  & $k$ & $\eta_2(k)$ & $
k^{-1}\eta_2(k)$\\ 
\hline
$1$ & \multicolumn{1}{r}{$-0.2552$} & \multicolumn{1}{r}{$-0.2552$} &  & $11$ & $4.9940$ & $0.4540$\\ 
$2$ & \multicolumn{1}{r}{$0.3433$} & \multicolumn{1}{r}{$0.1717$} &  & $12$
& $5.4957$ & $0.4580$\\ 
$3$ & \multicolumn{1}{r}{$0.8976$} & \multicolumn{1}{r}{$0.2992$} &  & $13$
& $5.9970$ & $0.4613$\\ 
$4$ & \multicolumn{1}{r}{$1.4308$} & \multicolumn{1}{r}{$0.3577$} &  & $14$
& $6.4978$ & $0.4641$\\ 
$5$ & \multicolumn{1}{r}{$1.9523$} & \multicolumn{1}{r}{$0.3905$} &  & $15$
& $6.9984$ & $0.4666$\\ 
$6$ & \multicolumn{1}{r}{$2.4667$} & \multicolumn{1}{r}{$0.4111$} &  & $16$
& $7.4988$ & $0.4687$\\ 
$7$ & \multicolumn{1}{r}{$2.9765$} & \multicolumn{1}{r}{$0.4252$} &  & $17$
& $7.9992$ & $0.4705$\\ 
$8$ & \multicolumn{1}{r}{$3.4834$} & \multicolumn{1}{r}{$0.4354$} &  & $18$
& $8.4994$ & $0.4722$\\ 
$9$ & \multicolumn{1}{r}{$3.9883$} & \multicolumn{1}{r}{$0.4431$} &  & $19$
& $8.9996$ & $0.4737$\\ 
$10$ & \multicolumn{1}{r}{$4.4916$} & \multicolumn{1}{r}{$0.4492$} &  & $
20 $ & $9.4997$ & $0.4750$\\ 
\hline
\end{tabular}}
\end{table}

\subsection{Simulated performances}\label{Sub 4.1}

In this section, we summarize selective simulation results to demonstrate the appealing properties, including both first-order and second-order, of the bounded variance point estimation methodologies that we provided as in \eqref{4.4}. We investigated a wide range of  scenarios in terms of the location and scale parameters of the negative exponential population (NExp), as well as the pre-specified parameters: the parameter $b^2$ for the bounded variance, the two parameters, $\rho$ and $k$, of $\mathcal{Q}(\rho,k)$ for using different sampling scheme. For brevity, we summarize the results from pseudo-random samples of a $NExp(5,2)$ population in Table \ref{Table 5}. We specify $b^2=0.0004,0.0001,0.000025$ and the optimal fixed sample size $n^{\ast}$ turned out to be $100,200$ and $400$ accordingly. To compare the number of sampling operations under different possible scenarios, we considered the combinations of $\rho=(1.0,0.8,0.5)$ and $k=(1,2,5)$. For each sampling scheme of $\mathcal{Q}(\rho,k)$, we included the average total final sample size $\bar{n}$ with the associated standard error $s(\bar{n})$, the difference between $\bar{n}$ and $n^{\ast}$ to be compared with the second-order efficiency term, $\rho^{-1}\eta_2(k)$, in Theorem \ref{Theorem 4.1} (ii), the $V(Y_{N:1})$ which should be close to the asymptotic variance as listed in Theorem \ref{Theorem 4.1} (iii),  as well as the average number of sampling operations $\bar{\varphi}$ to be compared with the expected number of sampling operations $E(\varphi)$ from \eqref{4.8}.

From \ref{Table 5}, it is obvious that $\bar{n}-n^*$ hangs tightly around each of its second-order approximation $\rho^{-1}\eta_2(k)$. From the eighth and ninth columns, one can also easily see that the average number of sampling operations needed, $\bar{\varphi}$, is very close to its theoretical value $E(\varphi)$. Moreover, the sampling operations for $\mathcal{Q}(\rho,k)$ are significantly reduced when $0<\rho<1$ and/or $k>1$, compared to Anscombe-Chow-Robbins purely sequential procedure $\mathcal{Q}(1,1)$ under the same $b$. And the reductions are approximately $100(1-k^{-1}\rho)\%$. The last column of Table \ref{Table 5} shows that the variance of the smallest observations are approximately $b^2$ across the board.

\begin{table}[t!]
\footnotesize
\captionsetup{font=footnotesize}
\caption{Simulations from $NExp(5,2)$ under $10,000$ runs implementing $\mathcal{Q}(\rho,k)$ from \eqref{4.4}}
\label{Table 5}\par
\vskip .2cm
\centerline{\tabcolsep=3truept\begin{tabular}{cccccccccc} \hline 
$n^{\ast}$ & $b$ & $\mathcal{Q}(\rho,k)$ & $\bar{n}$ & $s\left(\bar{n}\right) $ & $\bar{n}-n^{\ast}$ &$\rho^{-1}\eta_2(k)$& $\bar{\varphi}$ & $E(\varphi)$ &$V(Y_{N:1})$ \\ \hline
  {100} & {0.0004} & $\mathcal{Q}(1,1)$ & 99.7232 & 0.10214 & -0.2768 & -0.2552 & 95.723 & 97.952 & 0.000426 \\
          &       & $\mathcal{Q}(1,2)$ & 100.3324 & 0.10252 & 0.3324 & 0.3433 & 46.666 & 47.476 & 0.000421 \\
          &       & $\mathcal{Q}(1,5)$ & 101.9885 & 0.10333 & 1.9885 & 1.9523 & 17.198 & 17.190 & 0.000393 \\
          &       & $\mathcal{Q}(0.8,1)$ & 100.0483 & 0.11679 & 0.0483 & -0.3190 & 76.737 & 78.952 & 0.000441 \\
          &       & $\mathcal{Q}(0.8,2)$ & 100.8849 & 0.11555 & 0.8849 & 0.4291 & 37.654 & 38.476 & 0.000412 \\
          &       & $\mathcal{Q}(0.8,5)$ & 102.8252 & 0.11637 & 2.8252 & 2.4404 & 14.191 & 14.190 & 0.000387 \\
          &       & $\mathcal{Q}(0.5,1)$ & 99.4418 & 0.14938 & -0.5582 & -0.5104 & 46.721 & 48.952 & 0.000478 \\
          &       & $\mathcal{Q}(0.5,2)$ & 100.7288 & 0.14763 & 0.7288 & 0.6866 & 22.682 & 23.476 & 0.000428 \\
          &       & $\mathcal{Q}(0.5,5)$ & 103.7720 & 0.14944 & 3.7720 & 3.9046 & 8.177 & 8.190 & 0.000399 \\\hline
    {200} & {0.0001} & $\mathcal{Q}(1,1)$ & 199.8580 & 0.14426 & -0.1420 & -0.2552 & 195.858 & 197.952 & 0.000105 \\
          &       & $\mathcal{Q}(1,2)$ & 200.4932 & 0.14229 & 0.4932 & 0.3433 & 96.747 & 97.476 & 0.000101 \\
          &       & $\mathcal{Q}(1,5)$ & 202.1355 & 0.14324 & 2.1355 & 1.9523 & 37.227 & 37.190 & 0.000096 \\
          &       & $\mathcal{Q}(0.8,1)$ & 200.1291 & 0.15978 & 0.1291 & -0.3190 & 156.805 & 158.952 & 0.000100 \\
          &       & $\mathcal{Q}(0.8,2)$ & 201.0824 & 0.15891 & 1.0824 & 0.4291 & 77.733 & 78.476 & 0.000097 \\
          &       & $\mathcal{Q}(0.8,5)$ & 202.9941 & 0.15944 & 2.9941 & 2.4404 & 30.219 & 30.190 & 0.000095 \\
          &       & $\mathcal{Q}(0.5,1)$ & 199.2108 & 0.20546 & -0.7892 & -0.5104 & 96.605 & 98.952 & 0.000103 \\
          &       & $\mathcal{Q}(0.5,2)$ & 200.7224 & 0.20682 & 0.7224 & 0.6866 & 47.681 & 48.476 & 0.000100 \\
          &       & $\mathcal{Q}(0.5,5)$ & 204.0600 & 0.20489 & 4.0600 & 3.9046 & 18.206 & 18.190 & 0.000098 \\\hline
    {400} & {0.000025} & $\mathcal{Q}(1,1)$ & 399.8497 & 0.20026 & -0.1503 & -0.2552 & 395.850 & 397.952 & 0.000025 \\
          &       & $\mathcal{Q}(1,2)$ & 400.4404 & 0.19941 & 0.4404 & 0.3433 & 196.720 & 197.476 & 0.000025 \\
          &       & $\mathcal{Q}(1,5)$ & 402.0665 & 0.20082 & 2.0665 & 1.9523 & 77.213 & 77.190 & 0.000025 \\
          &       & $\mathcal{Q}(0.8,1)$ & 400.2258 & 0.22504 & 0.2258 & -0.3190 & 316.882 & 318.952 & 0.000025 \\
          &       & $\mathcal{Q}(0.8,2)$ & 401.0488 & 0.22446 & 1.0488 & 0.4291 & 157.718 & 158.476 & 0.000024 \\
          &       & $\mathcal{Q}(0.8,5)$ & 403.0076 & 0.22592 & 3.0076 & 2.4404 & 62.221 & 62.190 & 0.000025 \\
          &       & $\mathcal{Q}(0.5,1)$ & 399.5034 & 0.28495 & -0.4966 & -0.5104 & 196.752 & 198.952 & 0.000026 \\
          &       & $\mathcal{Q}(0.5,2)$ & 400.9348 & 0.28633 & 0.9348 & 0.6866 & 97.734 & 98.476 & 0.000026 \\
          &       & $\mathcal{Q}(0.5,5)$ & 403.9800 & 0.29034 & 3.9800 & 3.9046 & 38.198 & 38.190 & 0.000025 \\\hline
\end{tabular}}
\end{table}

\subsection{Real data analysis}\label{Sub 4.2}

In this section, we will implement the BVPE methodology $\mathcal{Q}({\rho,k})$ as per \eqref{4.4} on a real data set about survival times of a group of patients
suffering from head and neck cancer. This group of people were treated using a combination
of radiotherapy and chemotherapy. This data set has presented in multiple research articles such as \cite{Efron (1988)}, \cite{Shanker (2016)}, and most recently \cite{Zhuang and Bapat (2020)}.

It is fair to assume that the survival times data follows negative exponential distribution as it is claimed in \cite{Zhuang and Bapat (2020)}. 
Assuming that researchers in this study want to use the smallest observation of a sample data to estimate the location parameter $\mu$. And they also want to restrict the variance of the estimator to be $b^2=0.022$. Now we will implement the sampling scheme $\mathcal{Q}({\rho,k})$ into this investigation with a combination of $\rho= (1,0.8,0.5)$ and $k=(1,2,5)$.

For each sampling scheme $\mathcal{Q}_{\rho,k}$, the sampling procedure is the following: we randomize all of the observations but pretend these observations are not known to us. We start with $m=11$ observations, and proceed the sampling following the procedures as per \eqref{4.4}. We also assume the data coming in the order after randomization. We summarize the terminated sample size and the corresponding number of sampling operations in Table \ref{Table 6}. The columns are defined similarly to those as in \ref{Table 3}.

One can see from Table \ref{Table 6} that the terminated sample size ranges from $50$ to $62$, with the least number of observations when using the sampling scheme $\mathcal{Q}(\rho=1,k=1)$ and the most number of observations when using $\mathcal{Q}(\rho=0.5,k=5)$. Moreover, the sampling operations are reduced the most when using $\mathcal{Q}(\rho=0.5,k=5)$. And with the same $\rho$, larger $k$ means fewer sampling operations; and with the same $k$, smaller $\rho$ means fewer sampling operations. The last column recorded the point estimates, which were the minimum survival times observed in each sampling procedure, and all of them turned out to be 0.5. Finally, we should emphasize that all of these results were obtained from one single run, but we had indeed repeated similar implementations and there was little to no difference. 

\begin{table}[t!] 
\footnotesize
\captionsetup{font=footnotesize}
\caption{Terminated sample size associated with number of sampling operations using $\mathcal{Q}(\rho,k)$ as per \eqref{4.4}}
\label{Table 6}\par
\vskip .2cm
\centerline{\tabcolsep=3truept\begin{tabular}{cccc}
\hline
$\mathcal{Q}(\rho ,k)$ & $n_{\mathcal{Q}(\rho,k)}$ & $\varphi_{\mathcal{Q}(\rho,k)}$& $\hat{\mu}$\\ 
\hline
$\mathcal{Q}(1,1)$ & $50$ & $40$&$0.5$ \\ 
$\mathcal{Q}(1,2)$ & $51$ & $21$ &$0.5$\\ 
$\mathcal{Q}(1,5)$ & $51$ & $9$&$0.5$ \\ 
$\mathcal{Q}(0.8,1)$ & $53$ & $33$ &$0.5$\\ 
$\mathcal{Q}(0.8,2)$ & $54$ & $18$ &$0.5$\\ 
$\mathcal{Q}(0.8,5)$ & $58$ & $9$ &$0.5$\\ 
$\mathcal{Q}(0.5,1)$ & $54$ & $18$ &$0.5$\\ 
$\mathcal{Q}(0.5,2)$ & $54$ & $10$&$0.5$ \\ 
$\mathcal{Q}(0.5,5)$ & $62$ & $6$&$0.5$ \\ \hline
\end{tabular}}
\end{table}


\setcounter{equation}{0}
\section{Proofs}\label{Sect. 5}

Note that Theorem \ref{Theorem 1.1} and Theorem \ref{Theorem 2.1} follow from Theorem 2.4 of \cite{Woodroofe (1977)} immediately, \eqref{1.5} follows from Theorem \ref{Theorem 1.1} and the definition of $\varphi_{\mathcal{M}_0}$, \eqref{2.5} follows from Theorem \ref{Theorem 2.1}, and Theorem \ref{Theorem 3.1} is paraphrased from \citet[6.4.14]{Mukhopadhyay and de Silva (2009)}. In this section, therefore, we shall only prove Theorem \ref{Theorem 3.2} and Theorem \ref{Theorem 4.1}.

\subsection{Proof of Theorem 3.2 and (3.11)}

By the stopping rule defined in \eqref{3.10}, we have the following inequality:
\begin{equation}\label{5.1}
N_{\mathcal{P}(\rho,k)}\le \rho^{-1}(m+kT_{\mathcal{P}(\rho,k)})+1, 
\end{equation}
where 
\begin{equation*}
m+kT_{\mathcal{P}(\rho,k)}<m+k+\rho S_{m+k(T_{\mathcal{P}(\rho,k)}-1)}\sqrt{A/c}.
\end{equation*}
Therefore,
\begin{equation}\label{5.2}
\frac{N_{\mathcal{P}(\rho,k)}}{n^{\ast}}\le \frac{S_{m+k(T_{\mathcal{P}(\rho,k)}-1)}}{\sigma}+\frac{\rho^{-1}(m+k)+1}{n^{\ast}}\overset{P_{\mu,\sigma}}{\longrightarrow }1\text{ as }c\rightarrow 0.  
\end{equation}
On the other hand, we also have the inequality that
\begin{equation}\label{5.3}
N_{\mathcal{P}(\rho,k)}\ge \rho^{-1}(m+kT_{\mathcal{P}(\rho,k)})\ge S_{m+kT_{\mathcal{P}(\rho,k)}}\sqrt{A/c},
\end{equation}
from which we conclude 
\begin{equation}\label{5.4}
\frac{N_{\mathcal{P}(\rho,k)}}{n^{\ast}}\ge \frac{S_{m+kT_{\mathcal{P}(\rho,k)}}}{\sigma}\overset{P_{\mu,\sigma}}{\longrightarrow}1\text{ as } c\rightarrow 0.  
\end{equation}
Combining \eqref{5.2} and \eqref{5.4}, it is clear that
\begin{equation}\label{5.5}
N_{\mathcal{P}(\rho,k)}/n^{\ast}\overset{P_{\mu,\sigma}}{\longrightarrow}1\text{ as }c\rightarrow 0.  
\end{equation}

Note that for sufficiently small $c$, with limit operation given in \eqref{3.11}, we have (w.p.1) 
\begin{equation*}\label{5.6}
\frac{N_{\mathcal{P}(\rho,k)}}{n^{\ast}}\le \sigma^{-1}\sup\nolimits_{n\geq 2}S_{n}+2.  
\end{equation*}
Since it holds that
\begin{equation*}
E_{\mu,\sigma}^{2}[\sup\nolimits_{n\ge 2}S_{n}]\leq E_{\mu,\sigma}[(\sup\nolimits_{n\geq 2}S_{n})^{2}]\leq E_{\mu,\sigma}[\sup\nolimits_{n\ge 2}S_{n}^{2}],
\end{equation*}
and Wiener's (\citeyear{Wiener (1939)}) ergodic theorem leads to $E_{\mu,\sigma}[\sup\nolimits_{n\ge 2}S_{n}^{2}]<\infty$, combined with \eqref{5.5} it follows by the dominated convergence theorem that $$E_{\mu,\sigma}[N_{\mathcal{P}(\rho,k)}/n^{\ast}]\rightarrow 1 \text{ as } c\rightarrow 0.$$ Since $c=A\sigma^{2}/n^{\ast 2}$ from \eqref{3.3}, Theorem \ref{Theorem 3.2} (i) holds under the limit operations \eqref{3.11}.

\vspace{1em}

Recall that $T_{\mathcal{P}(\rho,k)}^{\prime}$ defined in \eqref{3.13} is of the same form with $t_{0}$ from \eqref{1.1}. Then, referring to (1.1) of \cite{Woodroofe (1977)} or Section A.4 of the Appendix in \cite{Mukhopadhyay and de Silva (2009)}, we have as $c\rightarrow 0$, for $m_{0}\geq 2$,
\begin{equation*}\label{5.7}
E_{\mu,\sigma }[T_{\mathcal{P}(\rho,k)}^{\prime}]=k^{-1}\rho n^{\ast}+\frac{1}{2}-\frac{3}{2k}-\frac{1}{2k}\sum_{n=1}^{\infty}n^{-1}E\left[\left\{\chi_{kn}^{2}-3kn\right\}^{+}\right] +o(1).  
\end{equation*}
Therefore, with $T_{\mathcal{P}(\rho,k)}^{\prime}=T_{\mathcal{P}(\rho,k)}+m$ w.p.1 and $m=m_{0}k+1$, we have
\begin{equation}\label{5.8}
E_{\mu,\sigma}[m+kT_{\mathcal{P}(\rho,k)}]=E_{\mu,\sigma }[1+kT_{\mathcal{P}(\rho,k)}^{\prime}]=\rho n^{\ast}+\eta_1(k)+o(1).  
\end{equation}
Putting together \eqref{3.15} and \eqref{5.8}, one obtains \eqref{3.16}. And under the limit operations \eqref{3.11}, Theorem \ref{Theorem 3.2} (ii) follows immediately from \eqref{5.8} and inequalities given in \eqref{5.1} and \eqref{5.3}.

\vspace{1em}

Next, we state the following lemmas to derive the desirable results in Theorem \ref{Theorem 3.2} (iii) and (iv).

\begin{lemma}\label{Lemma 5.1}
For the general sequential MRPE methodology $\mathcal{P}(\rho,k)$ given in \eqref{3.10}, under the limit operations \eqref{3.11}, for any arbitrary $0< \varepsilon <1$, with some $\gamma \ge 2,$
\begin{equation*}\label{4.9}
P_{\mu,\sigma}\left( N_{\mathcal{P}(\rho,k)}\le \varepsilon n^{\ast}\right)=O\left( n^{\ast^{-\frac{\gamma}{2r}}}\right).
\end{equation*}
\end{lemma}
\begin{proof}
Recall $\left\lfloor u\right\rfloor$ denotes the largest integer that is smaller than $u$ and we define:
\begin{equation*}\label{5.10}
t_{u}=\left\lfloor k^{-1}\rho \varepsilon n^{\ast}\right\rfloor +1.
\end{equation*}
It should be obvious that $0\le T_{\mathcal{P}(\rho,k)}\le t_{u}$. Then, the rate at which $P_{\mu,\sigma}\{N_{\mathcal{P}(\rho,k)}\le \varepsilon n^{\ast}\}$ may converge to zero under the limit operations \eqref{3.11} is given by
\begin{equation}\label{5.11}
\begin{split}
& P_{\mu,\sigma}\{N_{\mathcal{P}(\rho,k)}\le \varepsilon n^{\ast }\} \le P_{\mu,\sigma}\left\{ \rho^{-1}(m+kT_{\mathcal{P}(\rho,k)})\le \varepsilon n^{\ast}\right\}\\
\le & P_{\mu,\sigma}\left\{ S_{m+kt}\le \varepsilon \sigma \text{ for some }t\text{ such that }0\le t \le t_{u}\right\} \\
\le & P_{\mu,\sigma}\left\{\max\limits_{0\le t\le t_{u}}\left\vert S_{m+kt}-\sigma \right\vert \ge \left( 1-\varepsilon \right) \sigma \right\}\\
\le & \{\left( 1-\varepsilon \right) \sigma \}^{-\gamma}E_{\mu,\sigma}\left\vert S_{m}-\sigma \right\vert^{\gamma }=O\left( m^{-\frac{\gamma }{2}}\right) =O\left( n^{\ast-\frac{\gamma}{2r}}\right).
\end{split}
\end{equation}
For the justification of the last inequality in \eqref{5.11}, one may see \cite{Hu and Mukhopadhyay (2019)} and other sources for more details.
\end{proof}

\begin{lemma}\label{Lemma 5.2}
For the general sequential MRPE methodology $\mathcal{P}(\rho,k)$ given in \eqref{3.10}, under the limit operations \eqref{3.11},
\begin{equation*}
\begin{split}
(i) & \left( N_{\mathcal{P}(\rho,k)}-n^{\ast}\right)/n^{\ast 1/2}\overset{d}{\rightarrow}N(0,\frac{1}{2}\rho^{-1});\\ 
(ii) & \left( N_{\mathcal{P}(\rho,k)}-n^{\ast}\right)/N_{\mathcal{P}(\rho,k)}^{1/2}\overset{d}{\rightarrow}N(0,\frac{1}{2}\rho^{-1}).
\end{split}
\end{equation*}
\end{lemma}
\begin{proof}
First, we prove
\begin{equation}\label{5.12}
\frac{m+kT_{\mathcal{P}(\rho,k)}-\rho n^{\ast}}{\sqrt{\rho n^{\ast}}}\overset{d}{\rightarrow}N(0,\frac{1}{2}) \text{ as } c\rightarrow 0
\end{equation}
based on the inequalities that
\begin{equation*}\label{5.13}
\frac{\rho S_{m+kT_{\mathcal{P}(\rho,k)}}\sqrt{A/c}-\rho n^{\ast}}{\sqrt{\rho n^{\ast}}}\leq \frac{m+kT_{\mathcal{P}(\rho,k)}-\rho n^{\ast }}{\sqrt{\rho n^{\ast}}}\leq \frac{\rho S_{m+k(T_{\mathcal{P}(\rho,k)}-1)}\sqrt{A/c}-\rho n^{\ast}+m+k}{\sqrt{\rho n^{\ast}}}.
\end{equation*}
It is not hard to see that
\begin{equation*}\label{5.14}
\left( \left\lfloor \rho n^{\ast}\right\rfloor +1\right)
^{1/2}(S_{\left\lfloor \rho n^{\ast }\right\rfloor +1}/\sigma-1)\overset{d}{\rightarrow}N(0,\frac{1}{2}) \text{ as } c\rightarrow 0,
\end{equation*}
and the sequence $\{S_{\left\lfloor \rho n^{\ast }\right\rfloor +1}\}$ is {\it uniformly continuous in probability} (\citeauthor{Anscombe (1952)},\citeyear{Anscombe (1952)}). From previous results, we can easily show
\begin{equation*}\label{5.15}
\frac{m+kT_{\mathcal{P}(\rho,k)}}{\left\lfloor \rho n^{\ast }\right\rfloor+1}\overset{P_{\mu,\sigma}}{\longrightarrow}1 \text{ as } c\rightarrow 0.
\end{equation*}

Now Anscombe's (\citeyear{Anscombe (1952)}) random central limit theorem leads to 
\begin{equation*}\label{5.16}
\begin{split}
& \frac{\rho S_{m+kT_{\mathcal{P}(\rho,k)}}\sqrt{A/c}-\rho n^{\ast}}{\sqrt{\rho n^{\ast}}}\overset{d}{\rightarrow}N(0,\frac{1}{2}) \text{ and} \\ 
& \frac{\rho S_{m+k(T_{\mathcal{P}(\rho,k)}-1)}\sqrt{A/c}-\rho n^{\ast }+m+k}{\sqrt{\rho n^{\ast }}}\overset{d}{\rightarrow }N(0,\frac{1}{2}) \text{ as } c\rightarrow 0.
\end{split}
\end{equation*}
Hence, \eqref{5.12} holds. Next, with the inequalities that
\begin{equation}\label{5.17}
\rho ^{-\frac{1}{2}}\frac{m+kT_{\mathcal{P}(\rho,k)}-\rho n^{\ast }}{\sqrt{\rho n^{\ast}}}\le \frac{N_{\mathcal{P}(\rho ,k)}-n^{\ast }}{\sqrt{n^{\ast}}}\le \rho^{-\frac{1}{2}}\frac{m+kT_{\mathcal{P}(\rho,k)}-\rho n^{\ast}}{\sqrt{\rho n^{\ast}}}+\frac{1}{\sqrt{n^{\ast}}},
\end{equation}
Lemma \ref{Lemma 5.2} (i) follows immediately, and Slutsky's theorem provides Lemma \ref{Lemma 5.2} (ii) under the limit operations \eqref{3.11}. 
\end{proof}

\begin{lemma}\label{Lemma 5.3}
For the general sequential MRPE methodology $\mathcal{P}(\rho,k)$ given in \eqref{3.10}, under the limit operations \eqref{3.11},
\begin{equation*}
\left( N_{\mathcal{P}(\rho,k)}-n^{\ast}\right)^{2}/n^{\ast} \textit{ is uniformly integrable}.
\end{equation*}
\end{lemma}
\begin{proof}
In the light of \citet[Theorem 3.4]{Hu and Mukhopadhyay (2019)}, we can prove that $(\rho n^{\ast })^{-1}\left( m+kT_{\mathcal{P}(\rho,k)}-\rho n^{\ast}\right)^{2}$ is uniformly integrable for sufficient small $c\le c_{0}$ by choosing some $c_{0}(>0)$ appropriately. Therefore, under the limit operations \eqref{3.11}, we have Lemma \ref{Lemma 5.3} by applying the inequalities given in \eqref{5.17}.
\end{proof} 

\vspace{1em}

Now, Theorem \ref{Theorem 3.2} (iii) and (iv) follow from Lemmas \ref{Lemma 5.1}-\ref{Lemma 5.3}. Alternatively, appealing to nonlinear renewal theory, we can prove the same desired results in the spirit of \cite{Woodroofe (1977)}, as well. Many details are left out for brevity.

\subsection{Proof of Theorem 4.1 and (4.7)}

In the same fashion as we proved Theorem \ref{Theorem 3.2} (i), we have
\begin{equation}\label{5.18}
\frac{V_{m+kT}}{\sigma} \le \frac{N_{\mathcal{Q}(\rho,k)}}{n^*} \le \frac{V_{m+k(T_{\mathcal{Q}(\rho,k)}-1)}}{\sigma}+\frac{\rho^{-1}(m+k)+1}{n^*}.
\end{equation}
As $b \to 0$, the two bounds of the inequalities \eqref{5.18} both tend to 1 in probability, so $N_{\mathcal{Q}(\rho,k)}/n^* \overset{P_{\mu,\sigma}}{\longrightarrow} 1$. For sufficiently small $b$, with the limit operation \eqref{4.5}, we have (w.p.1) 
\begin{equation*}\label{5.19}
\frac{N_{\mathcal{Q}(\rho,k)}}{n^*} \le \sigma^{-1}\sup\nolimits_{n\ge2}V_n+2,
\end{equation*}
where $2(n-1)V_n/\sigma \sim \chi^2_{2n-2}$. Similarly, $E_{\mu,\sigma}[\sup\nolimits_{n\ge 2}V_{n}]<\infty$ follows from Wiener's (\citeyear{Wiener (1939)}) ergodic theorem so that
\begin{equation*}\label{5.20}
E_{\mu,\sigma}[{N_{\mathcal{Q}(\rho,k)}}/{n^*}] \to 1 \text{ as } b \to 0.
\end{equation*}
The proof of Theorem \ref{Theorem 4.1} (i) is complete. 

\vspace{1em}

Then, recall that $T_{\mathcal{Q}(\rho,k)}^{\prime}$ defined in \eqref{4.6} is of the same form with $t_{0}$ from \eqref{1.1}. Then, referring to (1.1) of \cite{Woodroofe (1977)} or Section A.4 of the Appendix in \cite{Mukhopadhyay and de Silva (2009)}, we have as $b\rightarrow 0$, for $m_{0}\geq 2$,
\begin{equation*}\label{5.21}
E_{\mu,\sigma }[T_{\mathcal{Q}(\rho,k)}^{\prime}]=k^{-1}\rho n^{\ast}+\frac{1}{2}-\frac{3}{2k}-\frac{1}{2k}\sum_{n=1}^{\infty}n^{-1}E\left[\left\{\chi_{2kn}^{2}-4kn\right\}^{+}\right] +o(1).  
\end{equation*}
Therefore, with $T_{\mathcal{Q}(\rho,k)}^{\prime}=T_{\mathcal{Q}(\rho,k)}+m$ w.p.1 and $m=m_{0}k+1$, we have
\begin{equation}\label{5.22}
E_{\mu,\sigma}[m+kT_{\mathcal{Q}(\rho,k)}]=E_{\mu,\sigma }[1+kT_{\mathcal{Q}(\rho,k)}^{\prime}]=\rho n^{\ast}+\eta_2(k)+o(1).  
\end{equation}
Putting together \eqref{4.7} and \eqref{5.22}, one obtains \eqref{4.8}. And under the limit operations \eqref{4.5}, Theorem \ref{Theorem 4.1} (ii) follows immediately from \eqref{5.22}.

\vspace{1em}

To evaluate the asymptotic variance in Theorem \ref{Theorem 4.1} (iii), we utilize the law of total variance and obtain
\begin{equation}\label{5.23}
\begin{split}
V_{\mu,\sigma}[Y_{N_{\mathcal{Q}(\rho,k)}:1}] &= E_{\mu,\sigma}[V(Y_{N_{\mathcal{Q}(\rho,k)}:1}|N_{\mathcal{Q}(\rho,k)})]+V_{\mu,\sigma}[E(Y_{N_{\mathcal{Q}(\rho,k)}:1}|N_{\mathcal{Q}(\rho,k)})]\\
& = \sum_{n=m}^{\infty}V(Y_{N_{\mathcal{Q}(\rho,k)}:1}|N_{\mathcal{Q}(\rho,k)}=n)P_{\mu,\sigma}(N_{\mathcal{Q}(\rho,k)}=n)+V_{\mu,\sigma}\left[\frac{\sigma^2}{N_{\mathcal{Q}(\rho,k)}}+\mu\right]\\
& = \sum_{n=m}^{\infty}\frac{\sigma^2}{n^2}P_{\mu,\sigma}(N_{\mathcal{Q}(\rho,k)}=n)+E_{\mu,\sigma}\left[\frac{\sigma^2}{N_{\mathcal{Q}(\rho,k)}^2}\right]-E_{\mu,\sigma}^2\left[\frac{\sigma}{N_{\mathcal{Q}(\rho,k)}}\right]\\
& = 2E_{\mu,\sigma}\left[\frac{\sigma^2}{N_{\mathcal{Q}(\rho,k)}^2}\right]-E_{\mu,\sigma}^2\left[\frac{\sigma}{N_{\mathcal{Q}(\rho,k)}}\right],
\end{split}
\end{equation}
since the event $\{N_{\mathcal{Q}(\rho,k)}=n\}$ depends on $V_n$ alone and is therefore independent of $Y_{n:1}$. 

Applying the Taylor theorem to expand $N_{\mathcal{Q}(\rho,k)}^{-j},j\ge1$ around $n^*$, we have
\begin{equation}\label{5.24}
N_{\mathcal{Q}(\rho,k)}^{-j} = n^{*-j} - j\lambda^{-j-1}(N_{\mathcal{Q}(\rho,k)}-n^*), 
\end{equation} 
where $\lambda$ is a random variable lying between $N_{\mathcal{Q}(\rho,k)}$ and $n^*$ Combining \eqref{5.23}, \eqref{5.24}, \eqref{4.3}, and Theorem \ref{Theorem 4.1} (ii) yields
\begin{equation*}
V_{\mu,\sigma}[Y_{N_{\mathcal{Q}(\rho,k)}:1}] = b^2 + O(b^3) = b^2+o(b^2),
\end{equation*} 
which completes the proof.


\setcounter{equation}{0}
\section{Concluding Remarks}\label{Sect. 6}

In this paper, we propose a broader and more general sequential sampling scheme with the motivation to save sampling operations while retaining efficiency. Following the idea of drawing multiple observations at-a-time sequentially to determine a preliminary sample, and then gathering the rest observations all in one batch, we demonstrate the MRPE and BVPE problems under the new sampling scheme as possible illustrations. Furthermore, the new sequential sampling scheme can be applied to deal with other statistical inference problems, including but not limited to: sequential analogues of Behrens-Fisher problems \citep[see e.g.][]{Robbins et al. (1967)}, fixed-width confidence intervals \citep[see e.g.][]{Hall (1983)}, ranking and selection \citep[see e.g.][]{Mukhopadhyay and Solanky (1994)}, bounded-risk point estimation \citep[see e.g.][]{Mukhopadhyay and Bapat (2018)}, treatment means comparison \citep[see e.g.][]{Mukhopadhyay et al (2022)}, etc.


Due to the appealing properties of our newly developed methodology and the reality of saving sampling operations substantially, it will be of great interest for more investigations on the recent problems that researchers have been working on. The list will keep going for a while, and we just list a few here to demonstrate the possible directions: (i) \cite{Schmegner and Baron (2004)} proposed a \textit{sequentially planned probability ratio test} (SPPRT) as a sequentially planned extension of the famous Wald's \textit{sequential probability ratio test} (SPRT); and (ii) \cite{Mukhopadhyay and Zhuang (2019)} worked on the two sample mean comparisons of normal distributions with unknown and unequal variances, where they developed both purely sequential and two-stage methodologies. Our sequential sampling design introduced in the paper can be directly applied to their problem settings, and is expected to save sampling operations significantly.


\end{document}